\documentclass[11pt]{article}

\usepackage{longtable}
\usepackage{microtype}
\usepackage[load-configurations=version-1]{siunitx} 
\usepackage{siunitx}
\usepackage{hyperref}
\hypersetup{colorlinks,
            linkcolor=blue,
            citecolor=blue,
            urlcolor=magenta,
            linktocpage,
            plainpages=false}
\usepackage[margin=1.5in]{geometry}

\usepackage{graphicx}
\usepackage{subfigure}
\usepackage{booktabs} 
\usepackage{amsmath}
\usepackage{amsfonts}
\usepackage{amsthm}
\usepackage[table]{xcolor}
\usepackage{tabularx}
\usepackage{multirow}
\usepackage{soul}
\usepackage{algorithm}
\usepackage{algorithmic}
\usepackage[algo2e]{algorithm2e}
\usepackage{bbm}
\usepackage{comment}
\usepackage[utf8]{inputenc}

\usepackage{tikz}
\usepackage[english]{babel}
\usepackage{bbold}

\newtheorem{Theorem}{Theorem}

\newtheorem{Corollary}{Corollary}

\newcommand{\R}{\mathbb{R}} 
\renewcommand{\H}{\mathcal{H}} 
\newcommand{\Sp}{\mathcal{S}} 
\newcommand{\E}{\mathrm{E}} 
\renewcommand{\d}{\mathrm{d}} 

\DeclareMathOperator*{\argmax}{arg\,max}

\newcommand{\matsuda}[1]{\textcolor{blue}{(TM: #1)}}
\newcommand{\wk}[1]{\textcolor{green!10!orange!90!}{(WX: #1)}}

\begin{document}

%

%

\title{A Stein Goodness-of-fit Test for Directional Distributions}
\author{Wenkai Xu$^1$ \and Takeru Matsuda$^{2,3}$}
\date{%
    $^1$Gatsby Unit of Computational Neuroscience\\%
    $^2$University of Tokyo\\
    $^3$RIKEN Center for Brain Science\\[2ex]%
}
\maketitle

\begin{abstract}
In many fields, data appears in the form of direction (unit vector) and usual statistical procedures are not applicable to such directional data. 
In this study, we propose non-parametric goodness-of-fit testing procedures for general directional distributions based on kernel Stein discrepancy.
Our method is based on Stein's operator on spheres, which is derived by using Stokes' theorem.
Notably, the proposed method is applicable to distributions with an intractable normalization constant, which commonly appear in directional statistics.
Experimental results demonstrate that the proposed methods control type-I error well and have larger power than existing tests, including the test based on the maximum mean discrepancy.
\end{abstract}

\section{INTRODUCTION}

In many applications, data is obtained in the form of directions and they are naturally identified with a vector on the unit hypersphere $\Sp^{d-1}=\{x  \in \R^d \mid \|x\|= 1\} \subset \R^d$. 
For example, wind direction is represented by a vector on the unit circle $\Sp^1 \subset \mathbb{R}^2$ \cite{genton2007blowing,hering2010powering}, while the protein structure is described by vectors on the unit sphere $\Sp^2 \subset \mathbb{R}^3$ \cite{protein}.
In addition, usual multivariate data in $\mathbb{R}^d$ is transformed to directional data by applying normalization, and such transformation is useful to analyze scale-invariant features.
For example, \cite{banerjee} transformed text document and gene expression data into directional data and applied model-based clustering.
Also, \cite{wang2017normface} showed that projecting face images to a unit hypersphere improves face recognition performance by convolutional neural networks.
Statistical methods for such directional data have been widely studied in the field of directional statistics \cite{mardia99}, and many statistical models of directional distributions have been proposed.
One characteristic feature of directional distributions is that they often involve an intractable normalization constant.
For example, the Fisher-Bingham distribution \cite{kent82} is defined by an unnormalized density
\[
    p(x \mid A,b) \propto \exp(x^{\top} A x + b^{\top} x), \quad x \in \Sp^{d-1},
\]
and its normalization constant is not represented in closed form.
Such intractable normalization constant makes statistical inferences for directional distributions computationally difficult.
While directional data are becoming increasingly important in many applications such as bioinformatics, meteorology, chronobiology, and text/image analysis, to the best of our knowledge,
goodness-of-fit testing for general directional distributions is not well established.

Several studies \cite{chwialkowski2016kernel,liu2016kernelized} have proposed kernel-based goodness-of-fit testing procedures for distributions on $\mathbb{R}^d$. 
These methods employ a model discrepancy measure called kernel Stein discrepancy (KSD), which is based on Stein's method \cite{barbour2005introduction,chen2010} and reproducing kernel Hibert space (RKHS) theory \cite{RKHSbook,muandet2017kernel}. 
Notably, the KSD test is applicable to unnormalized models, because it utilizes only the derivative of the logarithm of the density like score matching \cite{hyvarinen2005estimation}. 
This method is also applicable to model comparison \cite{jitkrittum2018informative,jitkrittum2017linear,kanagawa2019kernel}.
Recently, it has been extended to discrete distributions \cite{yang2018goodness} and point processes \cite{yang2019stein}.
On the other hand, applying Stein's method in the context of manifold structure is previously studied in \cite{barp2018riemannian} focusing on numerical integration problems for scalar functions and in \cite{liu2018riemannian} dealing with Bayesian inference on density functions.

In this study, we develop goodness-of-fit testing procedures for general directional distributions by extending kernel Stein discrepancy.
Our contributions are as follows. 
\begin{itemize}
\item We derive Stein's operator on the unit hypersphere $\Sp^{d-1}$ via Stokes' theorem and introduce directional kernel Stein discrepancy (dKSD).

\item We propose dKSD-based goodness-of-fit testing procedures for general directional distributions including unnormalized ones, which do not require to sample from the null distribution. 

\item We show that the proposed methods control type-I error well and have larger power than existing tests in simulation.
\end{itemize}

\paragraph{Paper Outline} 
We begin our presentation with a brief review of directional distributions and kernel Stein discrepancy on $\R^d$ in Section \ref{sec:background}. 
In Section 3, we derive Stein's operator on $\Sp^{d-1}$.
Then, after proposing directional kernel Stein discrepancy (dKSD) in Section 4, we develop goodness-of-fit testing procedures for directional distributions in Section 5. 
Experiment results are shown in Section \ref{sec:experiment} followed by conclusion in Section \ref{sec:conclusion}.

\section{BACKGROUND}\label{sec:background}

\begin{figure}[ht]
\label{fig:example-dist}
		\subfigure[Uniform]
		{\includegraphics[width=0.33\textwidth]{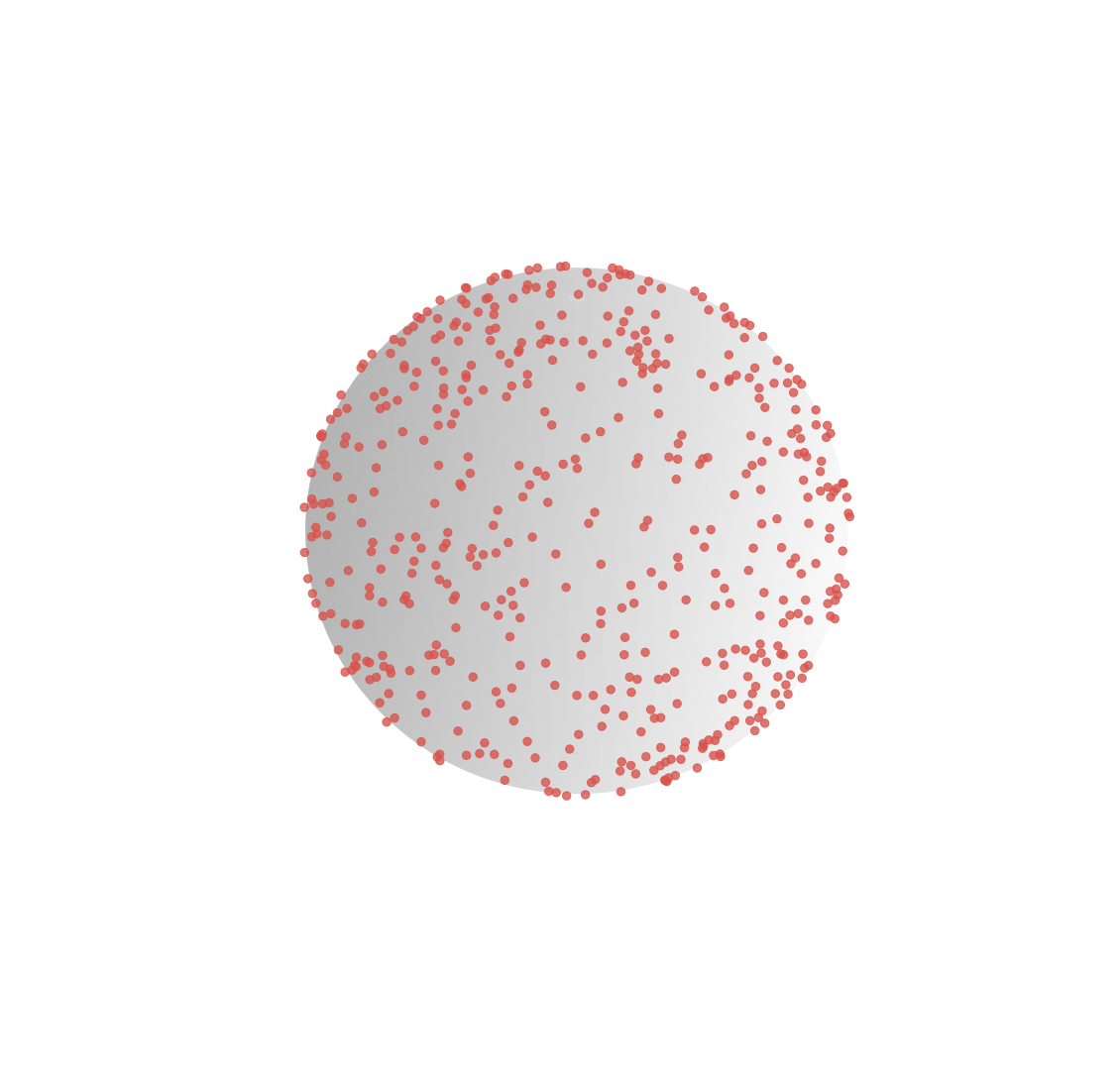}\label{fig:uniform-example}}\subfigure[von Mises-Fisher]
		{\includegraphics[width=0.33\textwidth]{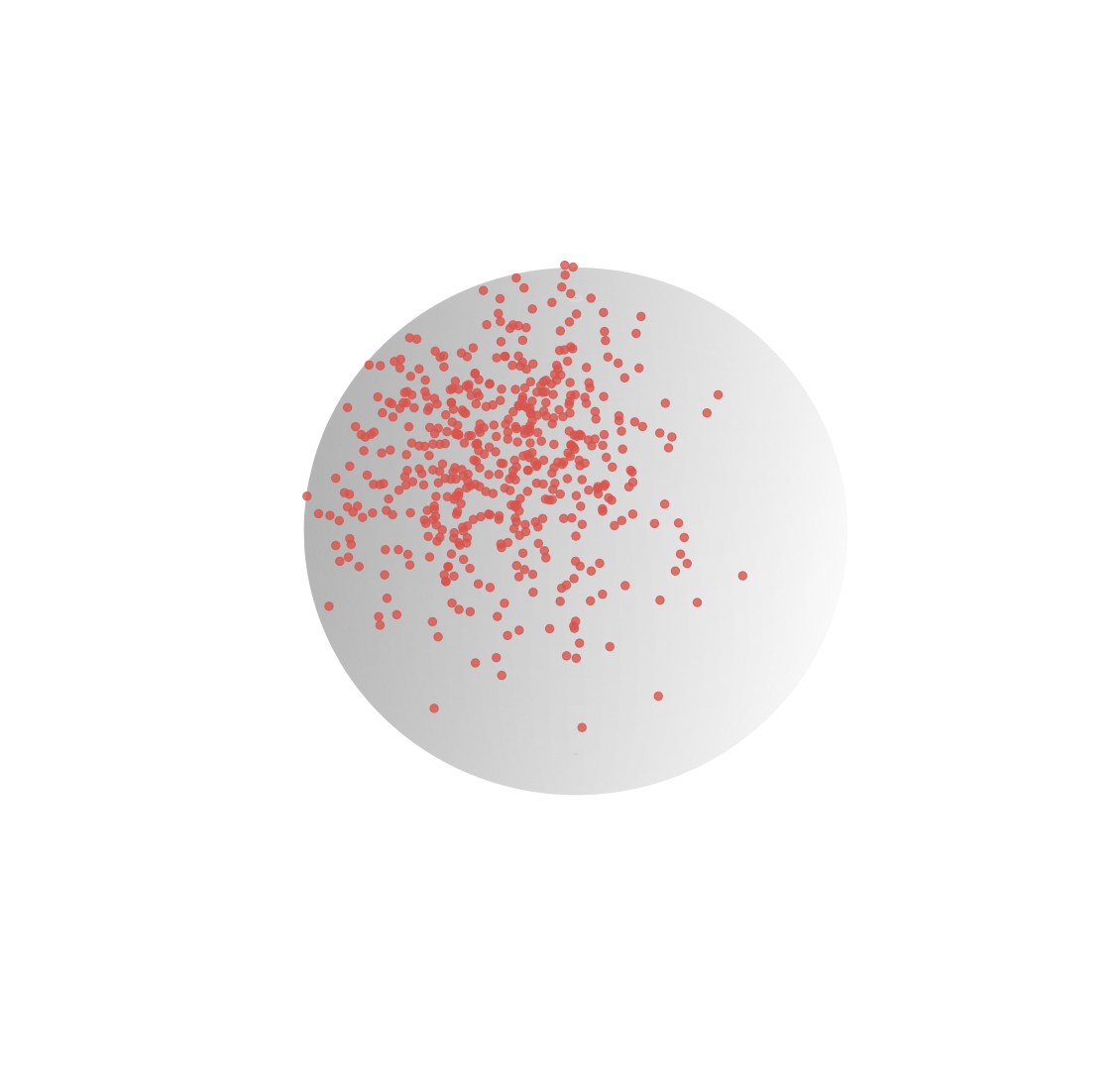}\label{fig:vmf-example}}\subfigure[Fisher-Bingham]
		{\includegraphics[width=0.33\textwidth]{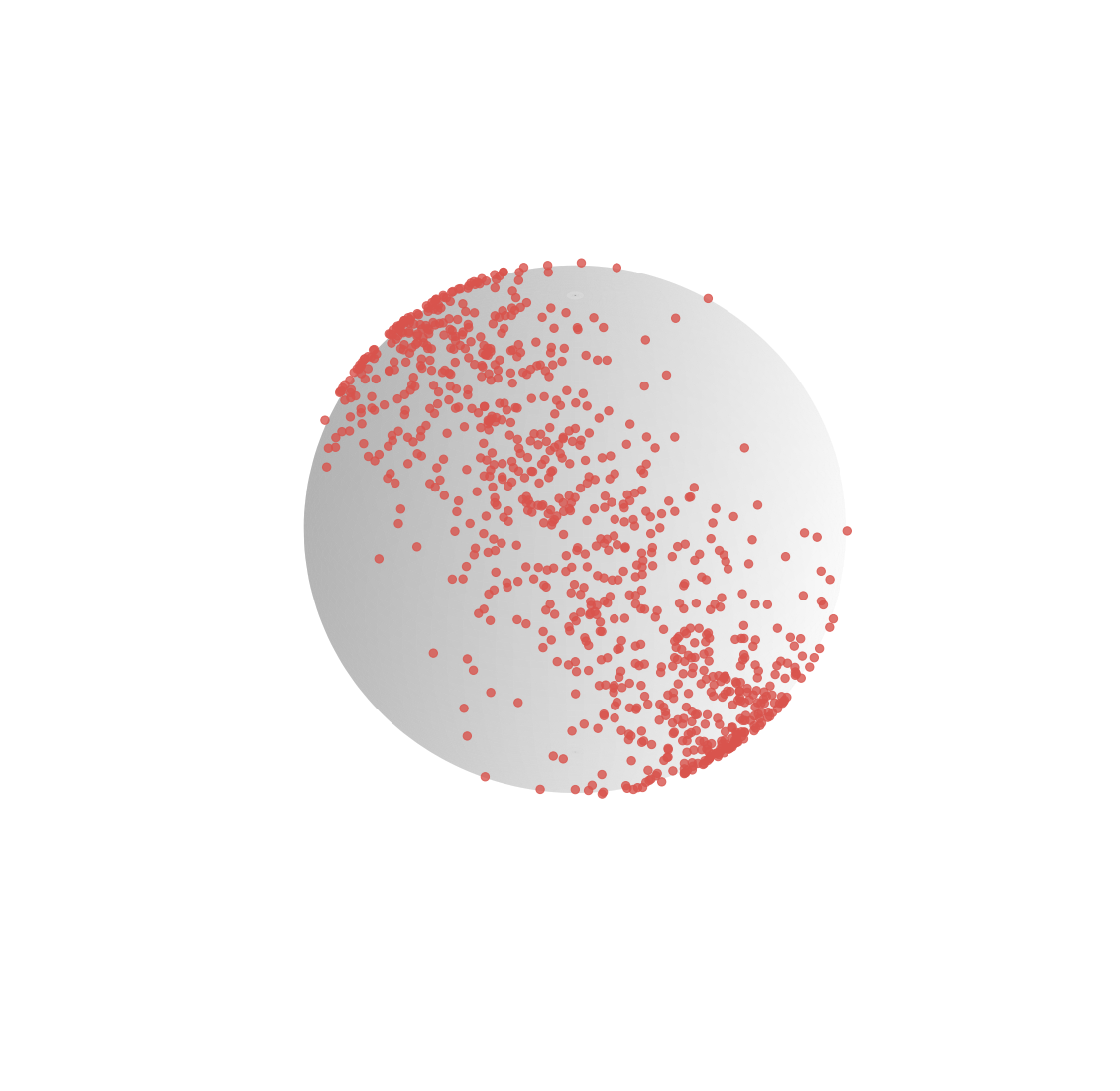}\label{fig:bingham-example}}
		\caption{Samples from three directional distributions on $\Sp^2$}
\end{figure}
\subsection{Directional Distributions}
Several distributions have been proposed for describing directional data on the unit hypersphere $\Sp^{d-1}=\{x  \in \R^d \mid \|x\|= 1\}$.
Here, we present two representative directional distributions: von Mises-Fisher and Fisher-Bingham.
Figure 1 shows samples from these distributions on $\Sp^2$.
See \cite{mardia99} for more detail.

In this paper, we define the probability density of directional distributions by taking the uniform distribution on $\Sp^{d-1}$ as base measure.
Namely, the density of the uniform distribution is $p(x) \equiv 1$.

The von Mises-Fisher (or von Mises when $d=2$) distribution is a directional counterpart of the isotropic Gaussian distribution on $\R^d$. 
Its density is given by
\begin{align}\label{eq:vonMises}
p(x \mid \mu,\kappa) = \frac{1}{C_{d}(\kappa)} \exp(\kappa \mu^{\top} x), 
\end{align}
for $x \in \Sp^{d-1}$, where $\mu \in \Sp^{d-1}$, $\kappa>0$,
\[
C_{d}(\kappa)=\frac{\kappa^{d/2-1}}{(2\pi)^{d/2}I_{d/2-1}(\kappa)},
\]
and $I_{v}$ is the modified Bessel function of the first kind and order $v$.
It is a unimodal distribution with peak at $\mu$ and  degree of concentration specified by $\kappa$.

The Fisher-Bingham (or Kent) distribution is an extension of the von Mises-Fisher distribution  \cite{kent82}.
Its density is given by
\begin{align}\label{eq:Fisher-Bingham}
p(x \mid A,b) = \frac{1}{Z(A,b)} \exp(x^{\top} A x + b^{\top} x), 
\end{align}
for $x \in \Sp^{d-1}$, where $A \in \mathbb{R}^{d \times d}$ is symmetric and $b \in \mathbb{R}^d$.
The normalization constant $Z(A,b)$ is not represented in closed form in general.


The goodness-of-fit test for general {directional distributions} is not well established, to the best of our knowledge. 
Tests for specific distributions such as uniform \cite{figueiredo2007comparison,garcia2018overview,mardia99} and von Mises-Fisher  \cite{figueiredo2012goodness,mardia84} cannot be readily extended to general directional distributions. 
Although \cite{boente2014goodness} proposed testing procedures based on the kernel density estimator, they are difficult to apply to unnormalized models such as the Fisher-Bingham distribution \eqref{eq:Fisher-Bingham}, because they require the normalization constant of the null model to calculate the $L^p$ test statistics. 

\subsection{Kernel Stein Discrepancy on  \texorpdfstring{$\R^d$}{Lg}
}
Here, we briefly review the goodness-of-fit testing with kernel Stein discrepancy on 
$\mathbb{R}^{d}$ by
\cite{chwialkowski2016kernel,liu2016kernelized}, which is inspired from \cite{gorham2015measuring,ley2017stein}.
See \cite{chwialkowski2016kernel,liu2016kernelized} for more detail.

Let $q$ be a smooth probability density on $\mathbb{R}^d$.
For a smooth function $f=(f_1,\dots,f_d):\mathbb{R}^d \to \mathbb{R}^d$, Stein's operator $\mathcal{T}_q$ is defined by 
\begin{align}
\mathcal{T}_q f(x)&=\sum_{i=1}^d \left( f_i(x) \frac{\partial}{\partial x^i} \log q(x) + \frac{\partial}{\partial x^i} f_i(x) \right).
\label{eq:steinRd}
\end{align}
From integration by parts on $\mathbb{R}^d$, we obtain the equality
\[
{\rm E}_q [\mathcal{T}_q f]=0
\]
under mild regularity conditions.
Since Stein's operator $\mathcal{T}_q$ depends on the density $q$ only through the derivatives of $\log q$, it does not involve the normalization constant of $q$, which is a useful property for dealing with unnormalized models \cite{hyvarinen2005estimation}. 

Let $\mathcal{H}$ be a reproducing kernel Hilbert space (RKHS) on $\mathbb{R}^d$ and $\mathcal{H}^d$ be its product.
By using Stein's operator, kernel Stein discrepancy (KSD) \cite{chwialkowski2016kernel,liu2016kernelized} between two densities $p$ and $q$ is defined as 
\begin{equation}
\textnormal{KSD}(p,q) =\sup_{\|f \|_{\mathcal{H}^d} \leq 1} \mathrm{E}_{p}[\mathcal{T}_q f]. 
\label{eq:ksd}
\end{equation}

It is shown that $\mathrm{KSD}(p,q) \geq 0$ and $\mathrm{KSD}(p,q) = 0$ if and only if $p=q$ under mild regularity conditions \cite{chwialkowski2016kernel}.
Thus, KSD is a proper discrepancy measure between densities.
After some calculation, $\mathrm{KSD}(p,q)$ is rewritten as
\begin{align}
\mathrm{KSD}^2(p,q) = {\rm E}_{x,\tilde{x} \sim p} [h_q(x,\tilde{x})], \label{eq:KSDequiv}
\end{align}
where $h_q$ does not involve $p$.

Now, suppose we have samples  $x_1,\dots,x_n$ from unknown density $p$ on $\mathbb{R}^d$.
Based on \eqref{eq:KSDequiv}, estimates of $\mathrm{KSD}^2(p,q)$ are obtained by using U-statistics or V-statistics.
These estimates can be used to test the hypothesis $H_0: p=q$.
The critical value is determined by  bootstrap based on the theory of U-statistics or V-statistics.
In this way, a general method of goodness-of-fit test on $\mathbb{R}^d$ is obtained, which is applicable to unnormalized models as well.

\section{STEIN'S OPERATOR ON \texorpdfstring{$\Sp^{d-1}$}{Lg}} \label{sec:stein-operator}





In this section, we derive Stein's operator for distributions on spheres. 
The derivation is based on Stokes' theorem, which is a fundamental theorem in differential geometry.

\subsection{Differential Forms and Stokes' Theorem}\label{subsec:stokes}
The original derivation of Stein's operator for distributions on $\mathbb{R}^d$ was based on integration by parts, in which the boundary term vanishes due to the decaying property of the probability density.
We need a different argument for spheres because its topology is different from $\mathbb{R}^d$.
Specifically, differential forms and Stokes' theorem are essential to discuss integration by parts on spheres.
Here, we briefly review these concepts.
See \cite{flanders,spivak2018calculus} for more detail and rigorous treatments.


Let $M$ be a $d$-dimensional closed manifold and take its local coordinate system $x^1,\dots,x^d$.
We introduce symbols ${\rm d}x^1,\dots,{\rm d}x^d$ and an associative and anti-symmetric operation $\wedge$ between them called the wedge product: ${\rm d}x^i \wedge {\rm d}x^j = -{\rm d}x^j \wedge {\rm d}x^i$.
Note that ${\rm d}x^i \wedge {\rm d}x^i = 0$.
Then, a $p$-form $\omega$ on $M$ ($0 \leq p \leq d$) is defined as
\[
    \omega = \sum_{i_1 \cdots i_p} f_{i_1 \cdots i_p} {\rm d} x^{i_1} \wedge \dots \wedge {\rm d} x^{i_p},
\]
where the sum is taken over all $p$-tuples $\{i_1, \cdots, i_p \} \subset \{1,\dots,d \}$ and each $f_{i_1 \cdots i_p}$ is a smooth function on $M$.
The exterior derivative ${\rm d} \omega$ of $\omega$ is defined as the $(p+1)$-form given by
\[
    {\rm d} \omega = \sum_{i_1 \cdots i_p} \sum_{i=1}^d \frac{\partial f_{i_1 \cdots i_p}}{\partial x^i} {\rm d} x^i  \wedge {\rm d} x^{i_1} \wedge \dots \wedge {\rm d} x^{i_p}.
\]

For another coordinate system $y^1,\dots,y^d$ on $M$, the differential form is transformed by
\[
    {\rm d} y^j = \sum_{i=1}^d \frac{\partial y^j}{\partial x^i} {\rm d} x^i.
\]

The integration of a $d$-form on a $d$-dimensional manifold is naturally defined like the usual integration on $\mathbb{R}^d$ and invariant with respect to the coordinate selection.
Correspondingly, the integration by parts formula on $\mathbb{R}^d$ is generalized in the form of Stokes' theorem.

\begin{Theorem}[Stokes' theorem]\label{thm:stoke's}
Let $\partial M$ be the boundary of $M$ and $\omega$ be a $(d-1)$-form on $M$.
Then,
\begin{align*}\label{eq:stokes}
\int_M \d\omega=\int_{\partial M}\omega.
\end{align*}
\end{Theorem}

In particular, since $\partial \Sp^{d-1}$ is empty, we obtain the following.

\begin{Corollary}\label{cor:stokes}
Let $\omega$ be a $(d-2)$-form on $\Sp^{d-1}$.
Then,
\begin{align}
\int_{S^{d-1}} \d\omega=0.
\end{align}
\end{Corollary}

Corollary \ref{cor:stokes} plays an important role in the derivation of Stein's operator on $\Sp^{d-1}$.



\subsection{Spherical Coordinate System}
In this paper, we use the spherical coordinate system $\theta=(\theta^1,\dots,\theta^{d-1})$ on $\Sp^{d-1}$ defined by
\begin{align}\label{eq:spherical}
\begin{pmatrix}\theta^1 \\ \theta^2 \\ \theta^3 \\ \vdots \\ \theta^{d-1}\end{pmatrix} \mapsto \begin{pmatrix}\cos \theta^{1}\\ \sin \theta^{1} \cos \theta^{2}\\ \sin \theta^{1} \sin \theta^{2} \cos \theta^3 \\ \vdots \\
\sin \theta^{1} \cdots \sin \theta^{d-1} \end{pmatrix} \in \Sp^{d-1}, 
\end{align}
where $(\theta^{1},\dots,\theta^{d-2}) \in [0, \pi)^{d-2}$ and  $\theta^{d-1} \in [0, 2\pi)$. 
In this coordinate system, the volume element \cite{flanders} is given by
\[
    {\rm d}S=J(\theta^1,\dots,\theta^{d-1}) {\rm d} \theta^1 \wedge \cdots \wedge {\rm d} \theta^{d-1},
\]
where
\[
    J(\theta^1,\dots,\theta^{d-1})=\sin^{d-2} (\theta^1) \sin^{d-3} (\theta^2) \cdots \sin (\theta^{d-2}).
\]
Note that $J(\theta^1)=1$ when $d=2$.
Since the surface area of $\Sp^{d-1}$ is $S_{d-1}=2\pi^{d/2}/\Gamma(d/2)$, the uniform distribution on $\Sp^{d-1}$ corresponds to the $(d-1)$-form $\eta$ on $\Sp^{d-1}$ given by
\[
    \eta = \frac{1}{S_{d-1}} J(\theta^1,\dots,\theta^{d-1}) {\rm d} \theta^1 \wedge \cdots \wedge {\rm d} \theta^{d-1}.
\]
By using this, the directional distribution on $\Sp^{d-1}$ with density $p$ is represented by the $(d-1)$-form $\omega$ given by
\[
    \omega = p \eta.
\]
Thus, expectation of a function $g$ with respect to $p$ is obtained by
\begin{align*}
    \mathrm{E}_p [g] &= \int_{\Sp^{d-1}} g \omega \\
    &= \frac{1}{S_{d-1}} \int_0^{2\pi} \int_0^{\pi} \cdots \int_0^{\pi} g(\theta) p(\theta) J(\theta) {\rm d} \theta^1 \cdots {\rm d} \theta^{d-1}.
\end{align*}

\subsection{Stein's Operator on \texorpdfstring{$\Sp^{d-1}$}{Lg}}
Now, we derive Stein's operator on $\Sp^{d-1}$ in the spherical coordinate.

\begin{Theorem}[Stein's operator on $\Sp^{d-1}$]\label{thm:directional-stein}
Let $p$ be a smooth probability density on $\Sp^{d-1}$.
For smooth functions $f_1,\dots,f_{d-1}:\Sp^{d-1} \to \mathbb{R}$, define a function $\mathcal{A}_{p}f:\Sp^{d-1} \to \mathbb{R}$ by
\begin{equation}\label{eq:directional-stein}
\mathcal{A}_{p}f = \sum_{i=1}^{d-1} \left( \frac{\partial f_i}{\partial {\theta}^i} + f_i \frac{\partial}{\partial {\theta}^i} \log (pJ) \right).
\end{equation}
Then,
\begin{align*}
{\rm E}_p [\mathcal{A}_{p} f] = 0.
\end{align*}
\end{Theorem}

\begin{proof}
Let ${\rm d}\theta^{(-i)}={\rm d}\theta^{i+1} \wedge \cdots \wedge {\rm d}\theta^{d-1} \wedge {\rm d}\theta^{1} \cdots \wedge {\rm d}\theta^{i-1}$ be a $(d-2)$-form on $\Sp^{d-1}$ for $i=1,\dots,\textnormal{d-1}$.
Consider a $(d-2)$-form $\omega$ on $\Sp^{d-1}$ defined by
\[
\omega = \sum_{i=1}^{{d-1}} f_i {\rm d} {\theta}^{(-i)}. 
\]
Then,
\begin{align*}
{\rm d}(pJ\omega) &= \sum_{i=1}^{d-1} \left( f_i \frac{\partial}{\partial {\theta}^i} (pJ) + pJ \frac{\partial f_i}{\partial {\theta}^i} \right) {\rm d} \theta^1 \wedge \cdots \wedge {\rm d} \theta^{{d-1}} \\
&=  (p J \mathcal{A}_{p} f) {\rm d} \theta^1 \wedge \cdots \wedge {\rm d} \theta^{{d-1}}. 
\end{align*}
From Corollary 1,
$
{\rm E}_p [\mathcal{A}_{p} f] = \int_{\Sp^{d-1}} {\rm d} (p J \omega) = 0.$
\end{proof}
Although Stein's operator on $\Sp^{d-1}$ has a similar form to the original Stein's operator on $\R^{d}$ in \eqref{eq:steinRd}, 
its derivation is different from the original one due to the topology of spheres.
Whereas the original derivation on $\R^{d}$ required {vanishing density} at the boundary, our derivation on $\Sp^{d-1}$ is free from such assumption.
Also note that, although we use the spherical coordinate system in this paper, we can derive Stein's operator in other coordinate systems as well.

\section{KERNEL STEIN DISCREPANCY ON \texorpdfstring{$\Sp^{d-1}$}{Lg}}\label{sec:dksd}
Based on Stein's operator on $\Sp^{d-1}$ in \eqref{eq:directional-stein}, we define the Stein discrepancy and its kernelized counterpart between two directional distributions via kernel mean embeddings, similar to \cite{chwialkowski2016kernel, liu2016kernelized}, which we call the directional kernel Stein discrepancy.


Let $\H$ be an RKHS on $\Sp^{d-1}$ with reproducing kernel $k$ and let $\H^{d-1}$ be its product. 
We define the directional kernel Stein discrepancy \textnormal{(dKSD)} by

\begin{equation}\label{eq:dksd}
\textnormal{dKSD}(p,q)=\sup_{\|f\|_{\mathcal{H}^{ d-1}} \leq 1} \mathbb{E}_{p} [\mathcal{A}_{q}f]
\end{equation}

Let $x$ and $\tilde{x}$ be points on $\Sp^{d-1}$ with spherical coordinates $\theta$ and $\tilde{\theta}$, respectively.
We identify the kernel function $k(x,\tilde{x})$ with a function of $\theta$ and $\tilde{\theta}$ through \eqref{eq:spherical} and take its derivatives.
For example, when $d=2$ and $k(x,\tilde{x})=\exp (\kappa x^{\top} \tilde{x})=\exp (\kappa \cos(\theta-\tilde{\theta}))$, we have
\[
\frac{\partial^{2}}{\partial \theta \partial \tilde{\theta}}k(x,\tilde{x})=\kappa (\cos(\theta-\tilde{\theta})-\kappa \sin^2(\theta-\tilde{\theta})) \exp(\kappa \cos(\theta-\tilde{\theta})).
\]
Let
\begin{align*}
    h_q(x,\tilde{x}) &= k(x,\tilde{x}) \sum_{i=1}^{d-1} \frac{\partial}{\partial \theta^i} \log (q(\theta)J(\theta)) \frac{\partial}{\partial \tilde{\theta}^i} \log (q(\tilde{\theta})J(\tilde{\theta}))\\
    &+ \sum_{i=1}^{d-1} \frac{\partial}{\partial \theta^i} \log (q(\theta)J(\theta)) \frac{\partial}{\partial \tilde{\theta}^i} k(x,\tilde{x}) \\
    &+ \sum_{i=1}^{d-1} \frac{\partial}{\partial \tilde{\theta}^i} \log (q(\tilde{\theta})J(\tilde{\theta})) \frac{\partial}{\partial {\theta}^i} k(x,\tilde{x}) \\
    &+ \sum_{i=1}^{d-1} \frac{\partial^2}{\partial {\theta}^i \partial \tilde{\theta}^i} k(x,\tilde{x}).
\end{align*}
Similarly to the original KSD \eqref{eq:KSDequiv}, dKSD is rewritten as follows.

\begin{Theorem}\label{th:dKSD-form}
Assume $p$ and $q$ are smooth densities on $\Sp^{d-1}$ and the reproducing kernel $k$ of $\mathcal{H}$ is a smooth function on $\Sp^{d-1} \times \Sp^{d-1}$. 
Then,
\begin{equation}\label{eq:dKSD-population}
\textnormal{dKSD}^2(p,q)=\mathbb{E}_{x,\tilde{x}\sim p}[h_{q}(x,\tilde{x})].
\end{equation}
\end{Theorem}
\begin{proof}
Since Stein's operator $\mathcal{A}_q$ is linear from \eqref{eq:directional-stein}, ${\rm E}_p[\mathcal{A}_q f]$ is a linear functional of $f \in \mathcal{H}^{d-1}$.
Then, from Riesz representation theorem, there uniquely exists $g=(g_1,\dots,g_{d-1}) \in \mathcal{H}^{d-1}$ such that ${\rm E}_p[\mathcal{A}_q f] = (f,g)_{\mathcal{H}^{d-1}}$.
By using the reproducing property of $\mathcal{H}$, we obtain
\begin{align}
    g_i(x) = {\rm E}_{\tilde{x} \sim p} \left[ k(x,\tilde{x}) \frac{\partial}{\partial \tilde{\theta}^i} \log (q(\tilde{\theta})J(\tilde{\theta})) + \frac{\partial}{\partial \tilde{\theta}^i} k(x,\tilde{x}) \right], \label{eq:def_g}
\end{align}
for $i=1,\dots,\textnormal{d-1}$.
Thus, the maximization in \eqref{eq:dksd} is attained by $f=g/\|g\|_{\mathcal{H}^{d-1}}$ and ${\rm dKSD}(p,q)=\| g \|_{\mathcal{H}^{d-1}}$.
Therefore, after straightforward calculations, we obtain \eqref{eq:dKSD-population}.
\end{proof}

Importantly, the function $h_q$ in \eqref{eq:dKSD-population} does not involve $p$.
Therefore, we can estimate $\textnormal{dKSD}^2(p,q)$ based on samples from $p$ and apply it to goodness-of-fit testing.

From the following theorem, $\textnormal{dKSD}^2(p,q)$ provides a proper discrepancy measure between directional distributions.
Let
\[
    L_i(x)=\frac{\partial}{\partial \theta^i} \log \frac{q(\theta)}{p(\theta)}, \quad i=1,\dots,\textnormal{d-1}.
\]

\begin{Theorem}\label{thm:characteristic}
Let $p$ and $q$ be smooth densities on $\Sp^{d-1}$. 
Assume the following:
\begin{itemize}
    \item The kernel $k$ is $C_0$-universal \textnormal{\cite[Definition 4.1]{carmeli2010vector}}.
    \item $\E_{x,\tilde{x} \sim p} h_p(x,\tilde{x}) < \infty$. 
    \item $\E_{p} \| L(x) \|^2 < \infty$. 
\end{itemize}
Then, $\textnormal{dKSD}^2(p,q)\geq 0$ and $\textnormal{dKSD}^2(p,q)=0$ if and only if $p=q$.
\end{Theorem}

\begin{proof}
From the proof of Theorem \ref{th:dKSD-form}, we have ${\rm dKSD}^2(p,q)=\|g\|_{\mathcal{H}^{d-1}}^2 \geq 0$, where $g=(g_1,\dots,g_{d-1})$ is defined as \eqref{eq:def_g}.
If $p=q$, then $\textnormal{dKSD}^2(p,q)=0$ from the definition \eqref{eq:dksd} and Theorem \ref{thm:directional-stein}. 
Conversely, if $\textnormal{dKSD}^2(p,q)=0$, then $g=0$, namely $g_i=0$ for $i=1,\dots,\textnormal{d-1}$.
Then, from $\log (q/p) = \log (qJ) - \log (pJ)$, we obtain
\begin{align*}
    {\rm E}_{\tilde{x} \sim p} \left[ L_i(\tilde{x}) k(x,\tilde{x}) \right] = g_i(x) - {\rm E}_{\tilde{x} \sim p} \left[ \mathcal{A}_p k(x,\tilde{x})  \right] = 0,
\end{align*}
for every $x$.
Since $k$ is $C_0$-universal, it implies $L_i=0$ \cite[Theorem 4.2b]{carmeli2010vector}.
Therefore, $\log (q/p)$ is constant on $\Sp^{d-1}$.
Since both $p$ and $q$ are densities on $\Sp^{d-1}$ that integrate to one, we obtain $p=q$.

\end{proof}




To apply dKSD for goodness-of-fit testing, we need to choose an RKHS on $\Sp^{d-1}$ that satisfies the conditions in Theorem \ref{thm:characteristic}.
In this paper, we use the RKHS generated by the von-Mises Fisher kernel:
\[
k(x,\tilde{x}) = \exp (\kappa x^{\top} \tilde{x}), \quad x,\tilde{x} \in \Sp^{d-1},
\]
where $\kappa>0$ is a concentration parameter that has a similar role to the band-width parameter in the Gaussian kernel.
Since both $x$ and $\tilde{x}$ have unit norm, their inner product $x^{\top} \tilde{x}$ is equal to the cosine of their angular separation. 
We discuss the method to choose $\kappa$ in Section 5.3.
See \cite{gneiting2013strictly} for general discussion on RKHS on $\Sp^{d-1}$.

\section{GOODNESS-OF-FIT TESTING VIA dKSD} \label{sec:goodness-of-fit}
In this section, we develop goodness-of-fit testing procedures based on dKSD.
Suppose $x_1,\cdots,x_n \sim p$ and we test $H_0: p=q$ with significance level $\alpha$.

\subsection{Test with U-statistics}
From \eqref{eq:dKSD-population}, an unbiased estimate of ${\rm dKSD}^2(p,q)$ is obtained in the form of U-statistics \cite{lee90}:
\begin{equation}\label{eq:u-stats}
{\textnormal{dKSD}}_u^2(p,q)=\frac{1}{n(n-1)}\sum_{i\neq j}h_{q}(x_{i},x_{j}).
\end{equation}
From the U-statistics theory \cite{lee90}, the asymptotic distribution of $\textnormal{dKSD}_u^2(p,q)$ is explicitly obtained as follows. 
Here, $\overset{d}{\to}$ denotes the convergence in distribution.

\begin{Theorem}\label{thm:u-stat-asymptotic}
Under the conditions in Theorem \ref{thm:characteristic}, the following statements hold.
\begin{enumerate}
\item Under $H_0: p = q$, the asymptotic distribution of $\textnormal{dKSD}_u^2(p,q)$ is
\begin{equation}\label{eq:null-dist}
n \cdot \textnormal{dKSD}_u^2(p,q)\overset{d}{\to}\sum_{j=1}^{\infty} c_{j}(Z_{j}^{2}-1),
\end{equation}
where $Z_{j}$ are i.i.d. standard Gaussian random variables and $c_{j}$ are the eigenvalues of the kernel $h_{q}(x,\tilde{x})$ under $p(\tilde{x})$:
\[
\int h_{q}(x,\tilde{x})\phi_{j}(\tilde{x})p(\tilde{x}){\rm d}\tilde{x} = c_{j}\phi_{j}(x), \quad \phi_{j}(x) \neq 0.
\]

\item Under $H_1: p\neq q$, the asymptotic distribution of $\textnormal{dKSD}_u^2(p,q)$ is 
\[
\sqrt{n}(\textnormal{dKSD}_u^2(p,q) - \textnormal{dKSD}^2(p,q))\overset{d}{\to}\mathcal{N}(0,\sigma_u^{2}),
\]
where $\sigma_u^{2}=\mathrm{Var}_{x\sim p}[\E_{\tilde{x}\sim p}[h_{q}(x,\tilde{x})]]\neq0$. 
\end{enumerate}
\end{Theorem}


The proof is essentially the same with Theorem 4.1 of \cite{liu2016kernelized}.
We employ Theorem \ref{thm:u-stat-asymptotic} for goodness-of-fit.
Namely, we generate bootstrap samples from an approximation of the null distribution \eqref{eq:null-dist} of $n \cdot \textnormal{dKSD}_u^2(p,q)$ and compare their $(1-\alpha)$ quantile with the realized value of $n \cdot \textnormal{dKSD}_u^2(p,q)$. 
To approximate the null, we truncate the infinite sum in \eqref{eq:null-dist} following \cite{gretton2009fast}: $\sum_{j=1}^{n} \hat{c}_{j}(Z_{j}^{2}-1)$, where $\hat{c}_j$ are eigenvalues of the $n \times n$ matrix ${H}$ with ${H}_{ij}=h(x_i,x_j)$ and $Z_1,\dots,Z_n$ are independent standard Gaussian random variables. 
The testing procedure is outlined in Algorithm \ref{alg:unbiased}.

\begin{algorithm}[tb]
   \caption{dKSD test via U-statistics (dKSDu)}
   \label{alg:unbiased}
\begin{algorithmic}[1]
\renewcommand{\algorithmicrequire}{\textbf{Input:}}
\renewcommand{\algorithmicensure}{\textbf{Objective:}}
\REQUIRE~~\\
    samples $x_1,\dots,x_n \sim p$\\
    null density $q$\\
    kernel function $k$\\
    test size $\alpha$\\
    bootstrap sample size $B$\\
\ENSURE~~
Test $H_0: p=q$ versus $H_1: p\neq q$.
\renewcommand{\algorithmicensure}{\textbf{Test procedure:}}
\ENSURE~~\\
\STATE Compute the U-statistics $\textnormal{dKSD}^2_u(p,q)$ via \eqref{eq:u-stats}.
\STATE Compute $n \times n$ matrix ${H}$ with ${H}_{ij}=h_q(x_i,x_j)$ and its eigenvalues $\hat{c}_1,\dots,\hat{c}_n$.
\FOR{$t=1:B$}
\STATE Sample $Z_1,\dots,Z_n \sim \mathcal{N}(0,1)$ independently.
\STATE Compute $S_t = \sum_{j=1}^n \hat{c}_j(Z_j^2-1)$.
\ENDFOR
\renewcommand{\algorithmicrequire}{\textbf{Output:}}
\STATE Determine the $(1-\alpha)$-quantile $\gamma_{1-\alpha}$ of $S_1,\dots,S_B$.
\REQUIRE~~\\
Reject $H_0$ if $n \cdot \textnormal{dKSD}^2_u(p,q) > \gamma_{1-\alpha}$; otherwise do not reject.
\end{algorithmic}
\end{algorithm}

\subsection{Wild Bootstrap Test with V-statistics}
Here, we propose another testing procedure with wild bootstrap adapted from \cite[Section 2.2]{chwialkowski2016kernel}, which is applicable even when observations $x_1,\dots,x_n \sim p$ are not independent. 
It is based on the V-statistics
\begin{equation}\label{eq:v-stats}
\textnormal{dKSD}_b^2(p,q)=\frac{1}{n^{2}}\sum_{i,j}h_{q}(x_{i},x_{j}).
\end{equation}
For each $t=1,\dots,B$, we sample uniform i.i.d. variables $U_1,\dots,U_n \sim \mathrm{U}[0,1]$, let $W_{0,t}=1$ and define
\begin{equation}\label{eq:bootstrap-weights}
W_{i,t} = \mathbbm{1}_{\{U_i > a_t\}}W_{i-1, t} - \mathbbm{1}_{\{U_i < a_t\}}W_{i-1, t},
\end{equation}
for $i=1,\dots,n$, where $a_t$ is the probability of sign change, which is set to $0.5$ when $x_1,\dots,x_n$ are independent. 

Then, wild bootstrap samples are given by
\begin{equation}\label{eq:bootstrap-null}
S_t = \frac{1}{n^2}\sum_{i,j} W_{i,t}W_{j,t}h(x_i,x_j), \quad t=1,\dots,n.
\end{equation}
We reject the null if the test statistic $\textnormal{dKSD}_b^2(p,q)$ in \eqref{eq:v-stats} exceeds the $(1-\alpha)$ quantile of $S_1,\dots,S_B$.
The testing procedure is outlined in Algorithm \ref{alg:wild}.

\begin{algorithm}[ht]
   \caption{dKSD test via wild bootstrap (dKSDv)}
   \label{alg:wild}
\begin{algorithmic}[1]
\renewcommand{\algorithmicrequire}{\textbf{Input:}}
\renewcommand{\algorithmicensure}{\textbf{Objective:}}
\REQUIRE~~\\
    samples $x_1,\dots,x_n \sim p$\\
    null density $q$\\
    kernel function $k$\\
    test size $\alpha$\\
    bootstrap sample size $B$\\
\ENSURE~~
Test $H_0: p=q$ versus $H_1: p\neq q$.
\renewcommand{\algorithmicensure}{\textbf{Test procedure:}}
\ENSURE~~\\
\STATE Compute the V-statistics ${\textnormal{dKSD}_b^2}(p,q)$ via \eqref{eq:v-stats}.
\FOR{$t=1:B$}
\STATE Sample $W_{1,t},\dots,W_{n,t}$ via \eqref{eq:bootstrap-weights}.
\STATE Compute $S_t$ by \eqref{eq:bootstrap-null}.
\ENDFOR
\renewcommand{\algorithmicrequire}{\textbf{Output:}}
\STATE Determine the $(1-\alpha)$-quantile $\gamma_{1-\alpha}$ of $S_1,\dots,S_B$.
\REQUIRE~~\\
Reject $H_0$ if ${\textnormal{dKSD}_b^2}(p,q) > \gamma_{1-\alpha}$; otherwise do not reject.
\end{algorithmic}
\end{algorithm}

\subsection{Kernel Choice}
In kernel-based testing, the performance is sensitive to the choice of kernel parameters such as the bandwidth parameter in Gaussian kernels \cite{gretton2012optimal}.
For the proposed dKSD tests with the von Mises-Fisher kernel $k(x,x') = \exp(\kappa x^{\top} x')$, the choice of concentration parameter $\kappa$ is crucial. 
Namely, if $\kappa$ is too small, the test magnifies any small difference between observed samples, and gives high type-I error. 
On the other hand, if $\kappa$ is too large, the test fails to detect the discrepancy between two different distributions.
Previous works \cite{chwialkowski2016kernel,gretton2012optimal,jitkrittum2018informative,jitkrittum2016interpretable,jitkrittum2017linear} proposed to choose the kernel parameter by maximizing the test power, which is defined as the probability of rejecting $H_0$ when it is false.
Here, we provide a method for choosing the kernel parameter by maximizing the test power of dKSDu.

We employ an approximation formula for the test power of dKSDu under $H_1: p \neq q$.
Since 
\[
D := \sqrt{n}\frac{ {\textnormal{dKSD}_u^2(p,q)}-\textnormal{dKSD}^2(p,q)}{\sigma_{u}}\overset{d}{\to}\mathcal{N}(0,1)
\]
from Theorem \ref{thm:u-stat-asymptotic}, we have
\begin{align*}
&\mathrm{Pr}_{H_1}(n \cdot {\textnormal{dKSD}_u^2(p,q)}>r) \\
=&\mathrm{Pr}_{H_1}\left( D>\frac{r}{\sqrt{n}\sigma_{H_1}}-\sqrt{n}\frac{\textnormal{dKSD}^2(p,q)}{\sigma_{u}} \right)\\
\approx& 1 - \Phi \left( \frac{r}{\sqrt{n}\sigma_{u}}-\sqrt{n}\frac{\textnormal{dKSD}^2(p,q)}{\sigma_{u}} \right), \label{eq:power}
\end{align*}
for large $n$ and fixed $r$, where $\Phi$ denotes the cumulative distribution function of the standard Gaussian distribution and $\sigma_{u}^2$ is defined in Theorem \ref{thm:u-stat-asymptotic}.
Following the argument in \cite{sutherland2016generative}, we use the approximation
\[
\frac{r}{\sqrt{n}\sigma_{u}}-\sqrt{n}\frac{\textnormal{dKSD}^2(p,q)}{\sigma_{u}} \approx -\sqrt{n}\frac{\textnormal{dKSD}^2(p,q)}{\sigma_{u}}.
\]
Thus, to maximize the test power, we choose $\kappa$ by
\[
\kappa^{\ast} = \argmax_{\kappa} \frac{\textnormal{dKSD}^2(p,q)}{\sigma_{u}}.
\]
In practice, we use part of the data to calculate ${{\textnormal{dKSD}_u^2}(p,q)}/{(\hat{\sigma}_{u}+\lambda)}$, where $\hat{\sigma}_{u}$ is an unbiased estimate of ${\sigma}_{u}$ and a regularization parameter $\lambda>0$ is added for numerical stability.
Then, we select $\kappa^{\ast}$ by grid search and apply the dKSD tests to the rest of the data.
In our experiments, this method had better testing performance than selecting the kernel parameter by the methods proposed in density estimation literature \cite{ garcia2013exact, garcia2013kernel,taylor2008automatic}.

\subsection{Test with Maximum Mean Discrepancy}

A proxy way to tackle the goodness-of-fit test on $\Sp^{d-1}$ is via the two-sample test with maximum mean discrepancy (MMD) \cite{gretton2007kernel}. 
Namely, to test whether $x_1,\dots,x_n$ is from density $q$, we draw samples $y_1,\dots,y_m$ from $q$ and determine whether $x_1,\dots,x_n$ and $y_1,\dots,y_m$ are from the same distribution. 
See \cite{gretton2007kernel} for details.
We compare the performance of the proposed dKSD tests with the MMD two-sample test in Section \ref{sec:experiment}.
Note that the MMD two-sample test requires to sample from the null distribution $q$, which can be computationally intensive for directional distributions especially in high dimension.
On the other hand, the proposed dKSD tests do not need samples from the null.



\section{EXPERIMENTAL RESULTS}\label{sec:experiment}

Here, we validate the proposed dKSD tests by simulation.
We employ the von Mises-Fisher kernel for both the dKSD tests and MMD two-sample test in Section 5.4.
The bootstrap sample size is set to $B=1000$.
The significance level is set to $\alpha = 0.01$.
In MMD two-sample test, we set $m=n$. 

\subsection{Circular Uniform Distribution}
First, we consider the circular ($d=2$) uniform distribution, for which several goodness-of-fit tests have been proposed such as Rayleigh test and Kuiper test \cite{mardia99}.
See Supplementary Material for details of Rayleigh test and Kuiper test.
We compare the proposed dKSD tests with these existing tests and MMD two-sample test.
We repeated 600 trials to calculate rejection rates. 

\begin{table}[ht]
\centering
\begin{tabular}{c|ccccc}
\toprule 
$n$ & Rayleigh & Kuiper & dKSDu & dKSDv & MMD \\
\midrule
30  &   0.006 & 0.010 & 0.011 & 0.007 & 0.013\\
50  &   0.015 & 0.011 & 0.015 & 0.015 & 0.016\\
100  &   0.010 & 0.011 & 0.008 & 0.011 & 0.030\\
200  &   0.015 & 0.018 & 0.010 & 0.015 & 0.013\\
\bottomrule
\end{tabular}
\caption{Type-I error of tests for the circular uniform distribution}
\label{tab:uni_size}
\end{table}

\begin{table}[ht]
\centering
\begin{tabular}{c|ccccc}
\toprule 
$n$ & Rayleigh & Kuiper & dKSDu  & dKSDv & MMD \\
\midrule
30  &   0.138 & 0.128& 0.560 & 0.338 & 0.133 \\
50  &   0.308& 0.267 & 0.750 & 0.898 & 0.317 \\
100  &  0.712& 0.667& 0.820 & 1.0  & 0.583 \\
200  &  0.980 & 0.962 & 0.900 & 1.0 & 0.900 \\
\bottomrule
\end{tabular}
\caption{Rejection rates for the circular uniform distribution under the von Mises distribution with $\kappa=0.5$}
\label{tab:uni_power}
\end{table}
\begin{table}[ht]
\centering
\begin{tabular}{c|ccccc}
\toprule 
$n$ & Rayleigh & Kuiper & dKSDu  & dKSDv  & MMD\\
\midrule
30  &   0.757& 0.731 & 0.650 & 0.831 & 0.600\\
50  &  0.957 & 0.940 & 0.750 & 1.0 &  0.833\\
100  &   1.0& 1.0&  0.833 &1.0 &  0.983\\
200  &   1.0 & 1.0 & 0.96 & 1.0 & 1.0\\
\bottomrule
\end{tabular}
\caption{Rejection rates for the circular uniform distribution under the von Mises distribution with $\kappa=1$}
\label{tab:uni_power2}
\end{table}

Table \ref{tab:uni_size} presents the rejection rate at the null.
The type-I errors of all tests are well controlled to the significance level $\alpha = 0.01$.

Tables \ref{tab:uni_power} and \ref{tab:uni_power2} present the rejection rate under the von Mises distribution with concentration $\kappa=0.5$ and $\kappa=1$, respectively.
The power of all tests increases with increasing $n$ or $\kappa$ and converges to one.
The dKSDv has the largest power.

\begin{figure*}[ht]
\label{fig:synthetic-problems}
	\begin{center}
        \subfigure[$d=3$, $\sigma=0$]
		{\includegraphics[width=0.25\textwidth]{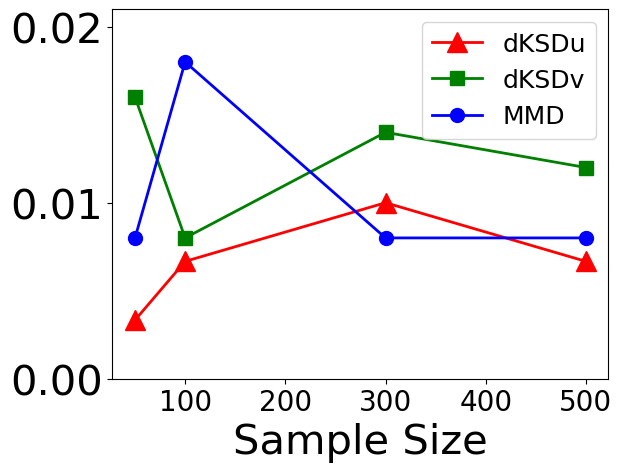}\label{fig:vmf-h0}}\subfigure[$d=3$, $\sigma=1$]
		{\includegraphics[width=0.24\textwidth]{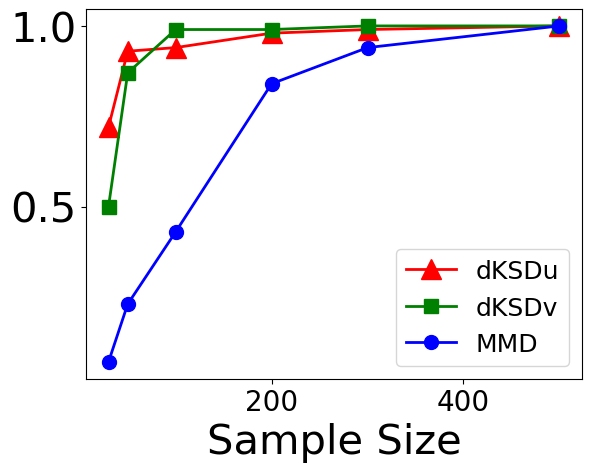}\label{fig:vmf-h1}}\subfigure[$d=3$, $n=200$]
		{\includegraphics[width=0.24\textwidth]{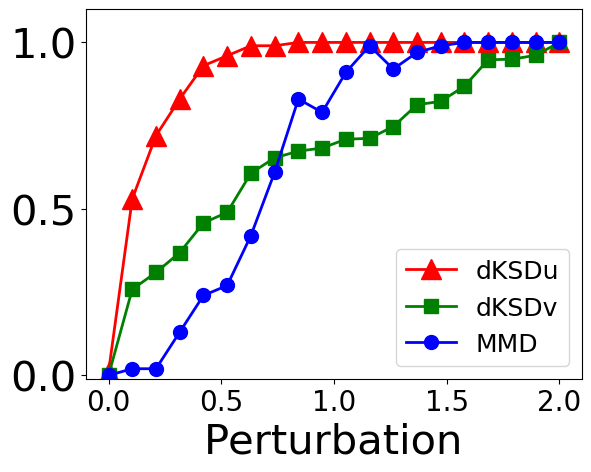}\label{fig:vmf-sig}}\subfigure[$n=200$, $\mu=d^{-1/2}\mathbf{1_d}$
		]
		{\includegraphics[width=0.25\textwidth]{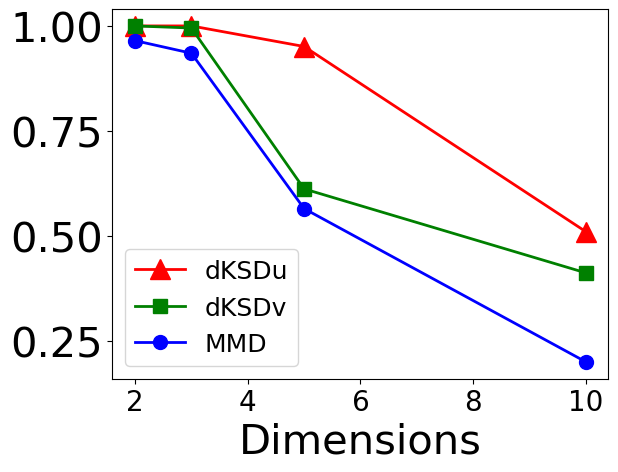}\label{fig:vmf-dim}}
		
		\subfigure[$d=3$, $\sigma=0$]
		{\includegraphics[width=0.25\textwidth]{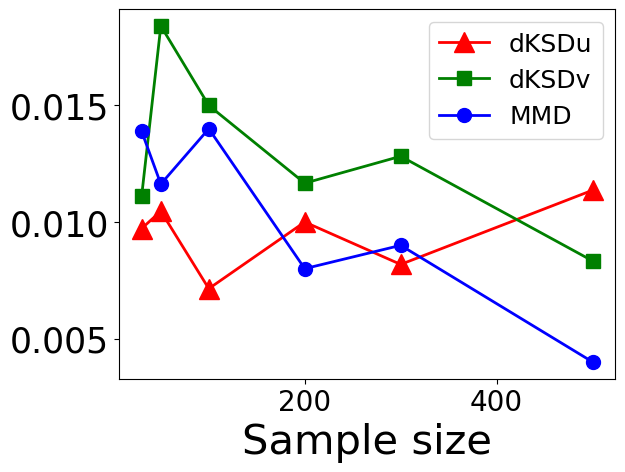}\label{fig:bingham-h0}}\subfigure[$d=3$, $\sigma=1$]
		{\includegraphics[width=0.24\textwidth]{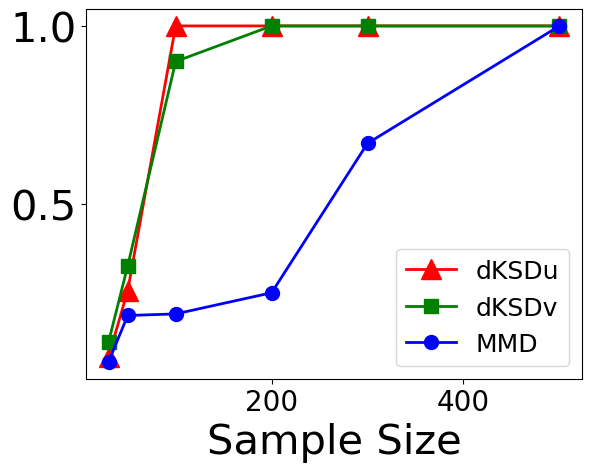}\label{fig:bingham-h1}}\subfigure[$d=3$, $n=200$]
		{\includegraphics[width=0.25\textwidth]{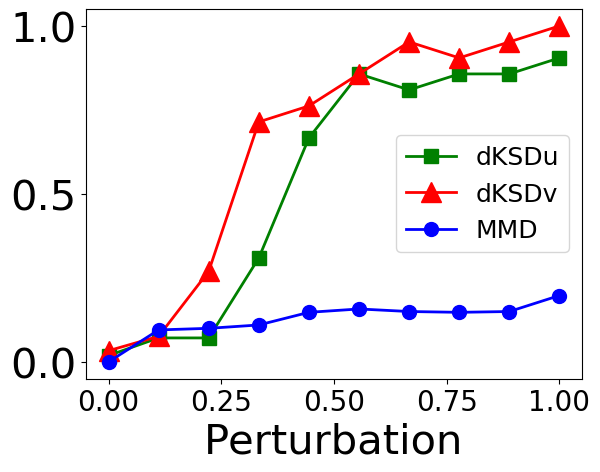}\label{fig:bingham-sig}}\subfigure[$n=200$, $\sigma=1$]
		{\includegraphics[width=0.25\textwidth]{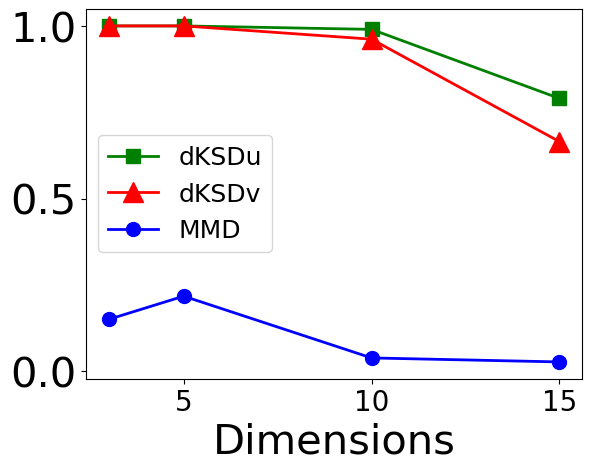}\label{fig:bingham-dim}}
		\caption{Rejection rates for (a)-(d) von Mises-Fisher; (e)-(h) Fisher-Bingham}		
	\end{center}
\end{figure*}

\subsection{von Mises-Fisher Distribution}

Next, we consider the von Mises-Fisher distribution ${\rm vMF}(\mu,\kappa)$ in \eqref{eq:vonMises}.
We compare the proposed dKSD tests with MMD two-sample test. 
We repeated 200 trials to calculate rejection rates. 

We set the null and alternative distributions to ${\rm vMF}(\mu_0,1)$ and ${\rm vMF}(\mu,1 + \sigma)$, respectively, where $\mu_0=(1,0,\dots,0) \in \Sp^{d-1}$, $\mu \in \Sp^{d-1}$ and $\sigma \geq 0$. 
We generated samples from the von Mises-Fisher distribution by using the methods proposed in \cite{jakob2012numerically,wood94}. 


Figure~\ref{fig:vmf-h0} plots the rejection rate under the null ($\mu=\mu_0$, $\sigma=0$) with respect to $n$ for $d=3$. 
The type-I errors of dKSD tests are well controlled to the significance level $\alpha=0.01$.   

Figure~\ref{fig:vmf-h1} plots the rejection rate with respect to $n$ for $d=3$, $\mu=\mu_0$ and $\sigma=1$. 
Both dKSDu and dKSDv have larger power than MMD two-sample test.

Figure~\ref{fig:vmf-sig} plots the rejection rate with respect to  $\sigma$ for $d=3$, $n=200$ and $\mu=\mu_0$. 
The dKSDu has the largest power and achieves almost 100\% power around $\kappa = 0.3$. 

Figure~\ref{fig:vmf-dim} plots the rejection rate with respect to $d$ for $n=200$, $\mu=(1/\sqrt{d})\mathbf{1_d}$ and $\sigma=0.5$, where $\mathbf{1_d}$ denotes the all one vector.
Although the test power decreases for higher dimension, dKSD tests have larger power than MMD two-sample test in all dimensions.

\subsection{Fisher-Bingham Distribution}
Finally, we consider the Fisher-Bingham distribution \eqref{eq:Fisher-Bingham}.
Here, we focus on the Fisher-Bingham distribution ${\rm FB}(A)$ that only includes second order terms:
\[
p({x} \mid A)\propto\exp({x}^{\top} A {x}), \quad x \in \Sp^{d-1},
\]
where $A\in \R^{d \times d}$ is symmetric. 
The normalization constant does not have closed form in general.
We compare the proposed dKSD tests with MMD two-sample test. 
We repeated 200 trials to calculate rejection rates. 

We set the null distribution to ${\rm FB}(A)$ with
\[
A_{ij} = \begin{cases} 2 & (i=j) \\ 1 & (i \neq j) \end{cases},
\]
and the alternative distribution to ${\rm FB}(A')$ with $A' = A + \sigma \textbf{1}_{d,d}$, where $\sigma \geq 0$ and $\textbf{1}_{d,d}$ denotes the $d \times d$ matrix with all entries one.
We generated samples from the Fisher-Bingham distribution via rejection sampling with angular central Gaussian proposals \cite{kent2013new, fallaize2016exact}. 


Figure~\ref{fig:bingham-h0} plots the rejection rate under the null ($\sigma=0$) with respect to $n$ for $d=3$.
The type-I errors of dKSD tests are approximately controlled to the significance level $\alpha=0.01$. 

Figure~\ref{fig:bingham-h1} plots the rejection rate with respect to $n$ for $d=3$ and $\sigma=1$.
The dKSD tests have larger power and achieve almost 100\% power around $n=100$.

Figure~\ref{fig:bingham-sig} plots the rejection rate with respect to $\sigma$ for $n=200$ and $d=3$.
Again, the dKSD tests have larger power and capture small perturbation. 

Figure~\ref{fig:bingham-dim} plots the rejection rate with respect to  $d$ for $n=200$ and $\sigma=1$. 
The dKSD tests attain almost 80\% power even when the dimension is as large as 15, whereas the power of the MMD two-sample test is smaller than 20\% for all dimensions.

Table~\ref{tab:bingham_time} presents the computational time for $d=3$.
The dKSD tests are more computationally efficient than MMD two-sample test.
The computational time of MMD two-sample test grows rapidly with the sample size $n$, because it requires to sample from the Fisher-Bingham distribution.

\begin{table}
\centering
\begin{tabular}{c|ccccc}
\toprule 
$n$ &  dKSDu  & dKSDv  & MMD\\
\midrule
30 & 0.005 & 0.009  &  0.091 \\
50 & 0.011 & 0.015 &  0.120  \\
100 & 0.027 & 0.030  &  0.180\\
200 &  0.096  & 0.105  &  0.379\\
300 & 0.227  & 0.238 & 0.704 \\
500 & 0.588 & 0.574 & 2.614  \\
\bottomrule
\end{tabular}
\caption{Computational time for Fisher-Bingham distribution (in seconds).}
\label{tab:bingham_time}
\end{table}



\section{CONCLUSION}\label{sec:conclusion}
In this study, we developed goodness-of-fit testing procedures for general directional distributions including unnormalized ones.
The proposed methods are based on an extension of Stein's operator and kernel Stein discrepancy.
Experimental results demonstrated that the proposed methods control type-I errors well and attain larger power than existing tests, without sampling from the null distribution.

Although we focused on the unit hypersphere $\Sp^{d-1}$ in this study, our derivation of Stein's operator and kernel Stein discrepancy is applicable to general manifolds as well.
It is an interesting future work to extend the proposed methods to general manifolds such as Stiefel manifolds and Grassmann manifolds.




\clearpage
\bibliographystyle{plain}
\bibliography{main}

\begin{thebibliography}{10}

\bibitem{banerjee}
A.~Banerjee, I.~S. Dhillon, J.~Ghosh, and S.~Sra.
\newblock Clustering on the unit hypersphere using von mises-fisher
  distributions.
\newblock {\em Journal of Machine Learning Research}, 6:1345--1382, 2005.

\bibitem{barbour2005introduction}
Andrew~D Barbour and Louis Hsiao~Yun Chen.
\newblock {\em An introduction to Stein's method}, volume~4.
\newblock World Scientific, 2005.

\bibitem{barp2018riemannian}
Alessandro Barp, Chris Oates, Emilio Porcu, and Mark Girolami.
\newblock A riemannian-stein kernel method.
\newblock {\em arXiv preprint arXiv:1810.04946}, 2018.

\bibitem{RKHSbook}
Alain Berlinet and Christine Thomas.
\newblock {\em Reproducing kernel Hilbert spaces in Probability and
  Statistics}.
\newblock Kluwer Academic Publishers, 2004.

\bibitem{boente2014goodness}
Graciela Boente, Daniela Rodriguez, and Wenceslao~Gonz{\'a}lez Manteiga.
\newblock Goodness-of-fit test for directional data.
\newblock {\em Scandinavian Journal of Statistics}, 41(1):259--275, 2014.

\bibitem{carmeli2010vector}
Claudio Carmeli, Ernesto De~Vito, Alessandro Toigo, and Veronica Umanit{\'a}.
\newblock Vector valued reproducing kernel hilbert spaces and universality.
\newblock {\em Analysis and Applications}, 8(01):19--61, 2010.

\bibitem{chen2010}
L.~H.~Y. Chen, L.~Goldstein, and Q.~M. Shao.
\newblock {\em Normal approximation by Stein's method}.
\newblock Springer, 2010.

\bibitem{chwialkowski2016kernel}
Kacper Chwialkowski, Heiko Strathmann, and Arthur Gretton.
\newblock A kernel test of goodness of fit.
\newblock In {\em International Conference on Machine Learning}, pages
  2606--2615, 2016.

\bibitem{fallaize2016exact}
Christopher~J Fallaize and Theodore Kypraios.
\newblock Exact bayesian inference for the bingham distribution.
\newblock {\em Statistics and Computing}, 26(1-2):349--360, 2016.

\bibitem{figueiredo2007comparison}
Adelaide Figueiredo.
\newblock Comparison of tests of uniformity defined on the hypersphere.
\newblock {\em Statistics \& probability letters}, 77(3):329--334, 2007.

\bibitem{figueiredo2012goodness}
Adelaide Maria~Sousa Figueiredo.
\newblock Goodness-of-fit for a concentrated von mises-fisher distribution.
\newblock {\em Computational Statistics}, 27(1):69--82, 2012.

\bibitem{flanders}
H.~Flanders.
\newblock {\em Differential Forms with Applications to the Physical Sciences}.
\newblock Dover, 1963.

\bibitem{garcia2013kernel}
Eduardo Garc{\'\i}a-Portugu{\'e}s, Rosa~M Crujeiras, and Wenceslao
  Gonz{\'a}lez-Manteiga.
\newblock Kernel density estimation for directional--linear data.
\newblock {\em Journal of Multivariate Analysis}, 121:152--175, 2013.

\bibitem{garcia2013exact}
Eduardo Garc{\'\i}a-Portugu{\'e}s et~al.
\newblock Exact risk improvement of bandwidth selectors for kernel density
  estimation with directional data.
\newblock {\em Electronic Journal of Statistics}, 7:1655--1685, 2013.

\bibitem{garcia2018overview}
Eduardo Garc{\'\i}a-Portugu{\'e}s and Thomas Verdebout.
\newblock An overview of uniformity tests on the hypersphere.
\newblock {\em arXiv preprint arXiv:1804.00286}, 2018.

\bibitem{genton2007blowing}
Marc Genton and Amanda Hering.
\newblock Blowing in the wind.
\newblock {\em Significance}, 4(1):11--14, 2007.

\bibitem{gneiting2013strictly}
Tilmann Gneiting et~al.
\newblock Strictly and non-strictly positive definite functions on spheres.
\newblock {\em Bernoulli}, 19(4):1327--1349, 2013.

\bibitem{gorham2015measuring}
Jackson Gorham and Lester Mackey.
\newblock Measuring sample quality with stein's method.
\newblock In {\em Advances in Neural Information Processing Systems}, pages
  226--234, 2015.

\bibitem{gretton2007kernel}
Arthur Gretton, Karsten Borgwardt, Malte Rasch, Bernhard Sch{\"o}lkopf, and
  Alex~J Smola.
\newblock A kernel method for the two-sample-problem.
\newblock In {\em Advances in neural information processing systems}, pages
  513--520, 2007.

\bibitem{gretton2009fast}
Arthur Gretton, Kenji Fukumizu, Zaid Harchaoui, and Bharath~K Sriperumbudur.
\newblock A fast, consistent kernel two-sample test.
\newblock In {\em Advances in neural information processing systems}, pages
  673--681, 2009.

\bibitem{gretton2012optimal}
Arthur Gretton, Dino Sejdinovic, Heiko Strathmann, Sivaraman Balakrishnan,
  Massimiliano Pontil, Kenji Fukumizu, and Bharath~K Sriperumbudur.
\newblock Optimal kernel choice for large-scale two-sample tests.
\newblock In {\em Advances in neural information processing systems}, pages
  1205--1213, 2012.

\bibitem{protein}
T.~Hamelryck, J.~T. Kent, and A.~Krogh.
\newblock Sampling realistic protein conformations using local structural bias.
\newblock {\em PLoS Comput. Biol.}, 2:e131, 2006.

\bibitem{hering2010powering}
Amanda~S Hering and Marc~G Genton.
\newblock Powering up with space-time wind forecasting.
\newblock {\em Journal of the American Statistical Association},
  105(489):92--104, 2010.

\bibitem{hyvarinen2005estimation}
Aapo Hyv{\"a}rinen.
\newblock Estimation of non-normalized statistical models by score matching.
\newblock {\em Journal of Machine Learning Research}, 6(Apr):695--709, 2005.

\bibitem{jakob2012numerically}
Wenzel Jakob.
\newblock Numerically stable sampling of the von mises-fisher distribution on
  $s^2$ (and other tricks).
\newblock {\em Interactive Geometry Lab, ETH Z{\"u}rich, Tech. Rep}, 2012.

\bibitem{jitkrittum2018informative}
Wittawat Jitkrittum, Heishiro Kanagawa, Patsorn Sangkloy, James Hays, Bernhard
  Sch{\"o}lkopf, and Arthur Gretton.
\newblock Informative features for model comparison.
\newblock In {\em Advances in Neural Information Processing Systems}, pages
  808--819, 2018.

\bibitem{jitkrittum2016interpretable}
Wittawat Jitkrittum, Zolt{\'a}n Szab{\'o}, Kacper~P Chwialkowski, and Arthur
  Gretton.
\newblock Interpretable distribution features with maximum testing power.
\newblock In {\em Advances in Neural Information Processing Systems}, pages
  181--189, 2016.

\bibitem{jitkrittum2017linear}
Wittawat Jitkrittum, Wenkai Xu, Zolt{\'a}n Szab{\'o}, Kenji Fukumizu, and
  Arthur Gretton.
\newblock A linear-time kernel goodness-of-fit test.
\newblock In {\em Advances in Neural Information Processing Systems}, pages
  262--271, 2017.

\bibitem{kanagawa2019kernel}
Heishiro Kanagawa, Wittawat Jitkrittum, Lester Mackey, Kenji Fukumizu, and
  Arthur Gretton.
\newblock A kernel stein test for comparing latent variable models.
\newblock {\em arXiv preprint arXiv:1907.00586}, 2019.

\bibitem{kent82}
J.~T. Kent.
\newblock The fisher–bingham distribution on the sphere.
\newblock {\em J. Royal. Stat. Soc. B}, 44:71--80, 1982.

\bibitem{kent2013new}
John~T Kent, Asaad~M Ganeiber, and Kanti~V Mardia.
\newblock A new method to simulate the bingham and related distributions in
  directional data analysis with applications.
\newblock {\em arXiv preprint arXiv:1310.8110}, 2013.

\bibitem{lee90}
A.~J. Lee.
\newblock {\em U-Statistics: Theory and Practice}.
\newblock CRC Press, 1990.

\bibitem{ley2017stein}
Christophe Ley, Gesine Reinert, Yvik Swan, et~al.
\newblock Stein’s method for comparison of univariate distributions.
\newblock {\em Probability Surveys}, 14:1--52, 2017.

\bibitem{liu2018riemannian}
Chang Liu and Jun Zhu.
\newblock Riemannian stein variational gradient descent for bayesian inference.
\newblock In {\em Thirty-Second AAAI Conference on Artificial Intelligence},
  2018.

\bibitem{liu2016kernelized}
Qiang Liu, Jason Lee, and Michael Jordan.
\newblock A kernelized stein discrepancy for goodness-of-fit tests.
\newblock In {\em International Conference on Machine Learning}, pages
  276--284, 2016.

\bibitem{mardia84}
K.~V. Mardia, D.~Holmes, and J.~Kent.
\newblock A goodness-of-fit test for the von mises-fisher distribution.
\newblock {\em J. Royal. Stat. Soc. B}, 46:72--78, 1984.

\bibitem{mardia99}
K.~V. Mardia and P.~E. Jupp.
\newblock {\em Directional Statistics}.
\newblock Wiley, New York, NY, 1999.

\bibitem{muandet2017kernel}
Krikamol Muandet, Kenji Fukumizu, Bharath Sriperumbudur, Bernhard
  Sch{\"o}lkopf, et~al.
\newblock Kernel mean embedding of distributions: A review and beyond.
\newblock {\em Foundations and Trends{\textregistered} in Machine Learning},
  10(1-2):1--141, 2017.

\bibitem{spivak2018calculus}
Michael Spivak.
\newblock {\em Calculus on manifolds: a modern approach to classical theorems
  of advanced calculus}.
\newblock CRC press, 2018.

\bibitem{sutherland2016generative}
Dougal~J Sutherland, Hsiao-Yu Tung, Heiko Strathmann, Soumyajit De, Aaditya
  Ramdas, Alex Smola, and Arthur Gretton.
\newblock Generative models and model criticism via optimized maximum mean
  discrepancy.
\newblock {\em arXiv preprint arXiv:1611.04488}, 2016.

\bibitem{taylor2008automatic}
Charles~C Taylor.
\newblock Automatic bandwidth selection for circular density estimation.
\newblock {\em Computational Statistics \& Data Analysis}, 52(7):3493--3500,
  2008.

\bibitem{wang2017normface}
Feng Wang, Xiang Xiang, Jian Cheng, and Alan~Loddon Yuille.
\newblock Normface: $l_2$ hypersphere embedding for face verification.
\newblock In {\em Proceedings of the 25th ACM international conference on
  Multimedia}, pages 1041--1049. ACM, 2017.

\bibitem{wood94}
A.~T.~A. Wood.
\newblock Simulation of the von mises fisher distribution.
\newblock {\em PLoS Comput. Biol.}, 23:157--164, 1994.

\bibitem{yang2018goodness}
Jiasen Yang, Qiang Liu, Vinayak Rao, and Jennifer Neville.
\newblock Goodness-of-fit testing for discrete distributions via stein
  discrepancy.
\newblock In {\em International Conference on Machine Learning}, pages
  5557--5566, 2018.

\bibitem{yang2019stein}
Jiasen Yang, Vinayak Rao, and Jennifer Neville.
\newblock A stein--papangelou goodness-of-fit test for point processes.
\newblock In {\em The 22nd International Conference on Artificial Intelligence
  and Statistics}, pages 226--235, 2019.

\end{thebibliography}







\clearpage
\appendix

\section{Uniformity test} 
We present Rayleigh test and Kuiper test for uniformity. 

\subsection{Rayleigh Test}
The test statistic of Rayleigh test is 
$$
R_n := \frac{2}{n} \left[\left( \sum_{i=1}^{n}\cos \theta_i \right)^2 + \left( \sum_{i=1}^{n} \sin \theta_i \right)^2 \right].
$$
Under the null, we have $R_n \sim \chi_2^{2}$. 
Therefore, the critical value is given by the quantile of chi-square distribution.
For example, if the significance level is set to $\alpha = 0.01$, then the critical value is $9.210$.

\subsection{Kuiper Test}
Kuiper test for uniformity is based on the cumulative distribution function (cdf).
The cdf of the uniform distribution is 
\[
F(\theta) = \frac{\theta}{2\pi}.
\]
We sort the samples to $0 \leq \theta_1 \leq \cdots \leq \theta_n \leq 2 \pi$ and compute 
$$
D^{+}_n := \sqrt{n} \sup_{\theta \in [0,2\pi)}\{F_n(\theta) - F(\theta)\} = \sqrt{n} \max_{1\leq i \leq n} \left( \frac{i}{n} - U_i \right),
$$
$$
D^{-}_n := \sqrt{n} \sup_{\theta \in [0,2\pi)}\{F(\theta) - F_n(\theta)\} = \sqrt{n} \max_{1\leq i \leq n} \left(U_i - \frac{i-1}{n} \right),
$$
where $U_i =\theta_i/(2\pi)$.
Then, the test statistic is defined as 
$$ V_n:= D_n^{+} + D_{n}^{-}.$$
The critical value is found in the statistical table.
For example, for the significance level $\alpha=0.01$, the critical value is $2.001$.

\end{document}